\documentclass[11pt,reqno]{amsart}
\usepackage{amsfonts}
\usepackage{mathrsfs}
\allowdisplaybreaks[4]
\usepackage{graphicx} 
\usepackage{epsfig}
\usepackage{amssymb}
\usepackage{amsmath}
\usepackage{cite}
\newtheorem{theorem}{Theorem}

\newtheorem{proposition}[theorem]{Proposition}

\newtheorem{lemma}[theorem]{Lemma}
\newcommand{\be}{\begin{equation}}
\newcommand{\ee}{\end{equation}}
\newcommand{\bea}{\begin{eqnarray}}
\newcommand{\eea}{\end{eqnarray}}
\newcommand{\ba}{\begin{array}}
\newcommand{\ea}{\end{array}}
\newcommand{\bean}{\begin{eqnarray*}}
\newcommand{\eean}{\end{eqnarray*}}

\newcommand{\pa}{\partial}

\begin{document}
\title{on the squared eigenfunction symmetry of the Toda lattice hierarchy}
\author{Jipeng Cheng\dag,Jingsong He$^{*}$ \ddag
 }
\dedicatory {  \dag \ Department of Mathematics,China University of
Mining and Technology, Xuzhou, Jiangsu 221116 ,
P.\ R.\ China\\
\ddag \ Department of Mathematics, Ningbo University, Ningbo ,
Zhejiang 315211, P.\ R.\ China }
\thanks{$^*$Corresponding author. Email:hejingsong@nbu.edu.cn; jshe@ustc.edu.cn.}
\begin{abstract}
The squared eigenfunction symmetry for the Toda lattice hierarchy is
explicitly constructed in the form of the Kronecker product of the
vector eigenfunction and the vector adjoint eigenfunction, which can
be viewed as the generating function for the additional symmetries
when the eigenfunction and the adjoint eigenfunction are the wave
function and the adjoint wave function respectively. Then after the
Fay-like identities and some important relations about the wave
functions are investigated, the action of the squared eigenfunction
related to the additional symmetry on the tau function is derived,
which is equivalent to the  Adler-Shiota-van Moerbeke (ASvM)
formulas.

\textbf{Keywords}:  squared eigenfunction symmetry, the Toda lattice
hierarchy, additional symmetry, Fay-like identities, ASvM formula.
\end{abstract}
\maketitle
\section{Introduction}
The squared eigenfunction
symmetry\cite{oevel1993,oevel1994,1OC98,2OC98}, also called ``ghost"
symmetry\cite{aratyn1998}, is a kind of symmetry generated by
eigenfunctions and adjoint eigenfunctions in the integrable system.
One important application lies in that by identifying the squared
eigenfunction symmetry with the usual flow of the integrable
hierarchy, one can get the corresponding symmetry
constraint\cite{oevel1994,2OC98,Cheng91,SS91,KS92,Cheng92,LorisWillox99,Tu2011}.
The other is its connection with the additional
symmetry\cite{aratyn1998,Cheng10,Li12}, which is the symmetry
depending explicitly on the space and time
variables\cite{Fokas1981,Chen1983,OS86,ASM95,D95,takasaki1996,Tu07,he2007,tian2010}.
By now, much work has done in the field of the squared eigenfunction
symmetry. For examples, 1) the squared eigenfunction symmetry for
the KP hierarchy is studied in its connection with the corresponding
additional symmetry in \cite{aratyn1998}; 2) by using the squared
eigenfunction symmetry to construct the new flow, the extended
integrable systems\cite{1Zeng08, 2Zeng08} are developed, which
contain the integrable equations with self-consistent sources; 3)
the squared eigenfunction symmetries for the BKP hierarchy and the
discrete KP hierarchy are systematically developed in \cite{Cheng10}
and \cite{Li12} respectively.

The Toda lattice equation\cite{toda}, as an important integrable
system, describes the motion of one-dimensional particles with
exponential interaction of neighbors, which plays significant role
in physics. The Toda lattice hierarchy was first introduced by Ueno
and Takasaki \cite{uenotaksasai} to generalize the Toda lattice
equations along the theory about the KP hierarchy\cite{DJKM}, which
is investigated in many aspects: such as interface growth
\cite{interface2000}, connection with Laplace-Darboux transformation
for general second order partial differential equations
\cite{nimmo1997}, connection with infinite dimensional Lie algebras
\cite{uenotaksasai,Mikhailov1981}, and its two extensions: extended
Toda hierarchy \cite{carlet03} and extended bigraded Toda
hierarchy\cite{carlet06,li2010,li2012}. However, the squared
eigenfunction symmetry of the Toda lattice hierarchy is not given
explicitly in literature.

In this paper, a handy form of the squared eigenfunction symmetry of
the Toda lattice hierarchy is given in the form of the Kronecker
product of the vector eigenfunctions and the vector adjoint
eigenfunctions. Then the relation with the additional symmetry of
the Toda lattice hierarchy is investigated: the particular squared
eigenfunction symmetry generated by the wave function and the
adjoint wave function can be viewed as the generating functions of
the additional symmetries for the Toda lattice
hierarchy\cite{ASM95,takasaki1996,cheng2011}. Next, in order to show
the action of the particular squared eigenfunction symmetry on the
tau function, which is actually the so-called Adler-Shiota-van
Moerbeke (ASvM) formulas \cite{ASM95,D95,L95}, the Fay-like
identities and some important relations about the wave functions are
studied. The Fay-like identities for the Toda lattice hierarchy
shows the algebraic properties for the tau
functions\cite{AvM99,teo2006,takasaki2007}. With the help of the
Fay-like identities, some relations about the wave functions are
derived in this paper. At last, upon the preparation above, the
action of the squared eigenfunction symmetry on the tau function is
obtained. By considering the connection between the squared
eigenfunction symmetry and the additional symmetry, another proof of
the ASvM formula is in fact showed.

This paper is organized in the following way. In Section 2, some
basic knowledge about the Toda lattice hierarchy is reviewed.
 Then, the squared eigenfunction symmetry of the Toda lattice hierarchy is constructed in Section 3.
 Next, in Section 4, the Fay-like identities and some relations about the wave functions are given.
 The action of the squared eigenfunction symmetry on the tau function is showed in Section 5.
 At last, in section 6, some conclusions and discussions are given.

\section{the Toda lattice hierarchy}
The Toda lattice hierarchy \cite{uenotaksasai,ASM95} is defined in
the Lax forms as
\begin{equation}\label{tlhierarchy}
    \pa_{x_n}L=[(L^n_1,0)_+,L]\ \ \ \text{and}\ \ \ \pa_{y_n}L=[(0,L_2^n)_+,L],\ \ \ n=1,2,\cdots
\end{equation}
with  $L$ be a pair of infinite matrices given by
\begin{equation}\label{laxoperator}
    L=(L_1,L_2)=\Big(\sum _{-\infty<i\leq 1}{\rm diag}[a_i^{(1)}(s)]\Lambda^i,\sum _{-1\leq i<\infty}{\rm diag}[a_i^{(2)}(s)]\Lambda^i\Big)\in
    \mathscr{D},
\end{equation}
where $\Lambda=(\delta_{j-i,1})_{i,j\in \mathbb{Z}}$, and
$a_i^{(1)}(s)$ and $a_i^{(2)}(s)$ are the functions of
$x=(x_1,x_2,\cdots)$ and $y=(y_1,y_2,\cdots)$, such that
$$a_1^{(1)}(s)=1 \ \ \ \text{and}\ \ \ a_1^{(2)}(s)\neq 0\ \ \ \forall s,$$
and the algebra
$$ \mathscr{D}=\{(P_1,P_2)\in {\rm gl}((\infty))\times {\rm gl}((\infty))\ | \ (P_1)_{ij}=0 \ {\rm for}\  j-i\gg0, \ (P_2)_{ij}=0 \ {\rm for} \  i-j\gg0\}$$
has the following splitting:
\begin{eqnarray*}
\mathscr{D}&=&\mathscr{D}_+ +\mathscr{D}_-,\\
\mathscr{D}_+&=&\{(P,P)\in \mathscr{D}\ | \ (P)_{ij}=0 \ {\rm for} \  |i-j|\gg 0\}=\{(P_1,P_2)\in \mathscr{D}\ | \ P_1=P_2\},\\
\mathscr{D}_-&=&\{(P_1,P_2)\in \mathscr{D}\ | \ (P_1)_{ij}=0\ {\rm
for}\ j\geq i, \ (P_2)_{ij}=0 \ {\rm for} \  i> j\},
\end{eqnarray*}
thus $(P_1,P_2)=(P_1,P_2)_+ +(P_1,P_2)_-$ can be given by
$$(P_1,P_2)_+=(P_{1u}+P_{2l},P_{1u}+P_{2l}),\ (P_1,P_2)_-=(P_{1l}-P_{2l},P_{2u}-P_{1u}),$$
where for a matrix $P$, $P_u$ and $P_l$ denote the upper (including
diagonal) and strictly lower triangular parts of $P$, respectively.

The Lax operator of the Toda lattice hierarchy (\ref{tlhierarchy})
can be expressed in terms of wave matrices $W=(W_1,W_2)$:
\begin{equation}\label{wavematrices}
    W_1(x,y)=S_1(x,y)e^{\xi(x,\Lambda)},\quad W_2(x,y)=S_2(x,y)e^{\xi(y,\Lambda^{-1})}
\end{equation}
as follows
\begin{equation}\label{laxwaveexpression}
    L=W(\Lambda,\Lambda^{-1})W^{-1}=S(\Lambda,\Lambda^{-1})S^{-1},
\end{equation}
where $S=(S_1,S_2)$ has the forms below
\begin{eqnarray}
S_1(x,y)=\sum_{i\geq 0} {\rm diag}[c_i(s;x,y)]
\Lambda^{-i},&&S_2(x,y)=\sum_{i\geq 0} {\rm diag}[c_i'(s;x,y)]
\Lambda^{i}, \label{hatwavematrices}
\end{eqnarray}
with $c_o(s;x,y)=1$ and $c_o'(s;x,y)\neq 0$ for any $s$, and
$\xi(x,\Lambda^{\pm1})=\sum_{n\geq1}x_n\Lambda^{\pm n}$. The wave
matrices evolve according to
\begin{eqnarray}
\pa_{x_n}S=-(L_1^n,0)_-S,&& \pa_{y_n}S=-(0,L_2^n)_-S, \label{hatwavematricesequation}\\
\pa_{x_n}W=(L_1^n,0)_+W,&& \pa_{y_n}W=(0,L_2^n)_+W.
\label{wavematricesequation}
\end{eqnarray}

The vector wave functions $\Psi=(\Psi_1,\Psi_2)$ are defined in the
following way
\begin{equation}\label{vectorwavefunction}
    \Psi_i(x,y;z)=(\Psi_i(n;x,y;z))_{n\in\mathbb{Z}}=W_i(x,y)\chi(z),
\end{equation}
where $\chi(z)=(z^i)_{i\in\mathbb{Z}}$, while the adjoint wave
functions are defined by
\begin{equation}\label{adjointwavefunction}
    \Psi_i^*(x,y;z)=(\Psi_i^*(n;x,y;z))_{n\in\mathbb{Z}}:=(W_i(x,y)^{-1})^T\chi^*(z)
\end{equation}
with $\chi^*(z)=\chi(z^{-1})$  and $T$ refers to the matrix
transpose. Thus from (\ref{laxwaveexpression}), we know
\begin{equation}\label{laxactonwave}
    L\Psi=(z,z^{-1})\Psi.
\end{equation}
And further the wave functions satisfy the following equations
\begin{eqnarray}
\pa_{x_n}\Psi=(L_1^n,0)_+\Psi,&&
\pa_{y_n}\Psi=(0,L_2^n)_+\Psi.\label{vectorwavefunctionequation}
\end{eqnarray}

If vector functions $q=(q(n;x,y))_{n\in\mathbb{Z}}$ and
$r=(r(n;x,y))_{n\in\mathbb{Z}}$ satisfy
\begin{eqnarray}
\pa_{x_n}q=(L_1^n)_uq,&&\pa_{y_n}q=(L_2^n)_lq,\nonumber\\
\pa_{x_n}r=-(L_1^n)^T_uq,&&\pa_{y_n}r=-(L_2^n)^T_lr,\label{qrdef}
\end{eqnarray}
we call them \textbf{vector eigenfunction} and \textbf{vector
adjoint eigenfunction} for the Toda lattice hierarchy respectively.
Obviously, the wave functions $\Psi_1$ and $\Psi_2$ are
eigenfunctions, and the adjoint wave functions $\Psi_1^*$ and
$\Psi_2^*$ are the adjoint eigenfunctions.

From (\ref{wavematricesequation}), the bilinear relation for the
Toda lattice hierarchy can be easily got
\begin{equation}\label{bilinearrelation}
   W_1(x,y)W_1(x',y')^{-1}=W_2(x,y)W_2(x',y')^{-1}
\end{equation}
for any $x,x'$ and $y,y'$. In order to rewrite
(\ref{bilinearrelation}) in terms of the (adjoint) wave functions,
the following lemma\cite{ASM95} is needed.
\begin{lemma}\label{lemma}
Given two operators $U=(U_1,U_2)$, $V=(V_1,V_2)\in \mathscr{D}$
depending on $x$ and $y$, one has
\begin{eqnarray}
U_1V_1&=&{\rm res}_{z}\frac{1}{z}\big(U_1\chi(z)\big)\otimes \big(V_1^T\chi^*(z)\big),\label{lemma1}\\
U_2V_2&=&{\rm res}_{z}\frac{1}{z}\big(U_2\chi(z^{-1})\big)\otimes
\big(V_2^T\chi^*(z^{-1})\big),\label{lemma2}
\end{eqnarray}
where ${\rm res}_z\sum_ia_iz^i=a_{-1}$ and $(A\otimes
B)_{ij}=A_iB_j$.
\end{lemma}
Therefore, according to
(\ref{vectorwavefunction})(\ref{adjointwavefunction}), the bilinear
relation (\ref{bilinearrelation}) is equivalent to the following
residue formula,
\begin{equation}\label{wavebilinearmatrix}
    {\rm res}_{z}\frac{1}{z}\Psi_1(x,y;z)\otimes \Psi_1^*(x',y';z)={\rm res}_{z}\frac{1}{z}\Psi_2(x,y;z^{-1})\otimes
    \Psi_2^*(x',y';z^{-1}),
\end{equation}
which further can be showed as
\begin{equation}\label{wavebilinearcomponent}
    {\rm res}_{z}\frac{1}{z}\Psi_1(s;x,y;z)\Psi_1^*(s';x',y';z)={\rm
    res}_{z}\frac{1}{z}\Psi_2(s;x,y;z^{-1})\Psi_2^*(s';x',y';z^{-1}).
\end{equation}
Using the bilinear identity (\ref{wavebilinearcomponent}) above,
Ueno and Takasaki \cite{uenotaksasai} proved that there exists a tau
function $\tau(s,x,y)$ such that
\begin{eqnarray}
\Psi_1(s;x,y;z)&=&\frac{\tau(s;x-[z^{-1}],y)}{\tau(s;x,y)}e^{\xi(x,z)}z^s,\nonumber\\
\Psi_2(s;x,y;z)&=&\frac{\tau(s+1;x,y-[z])}{\tau(s;x,y)}e^{\xi(y,z^{-1})}z^s,\nonumber\\
\Psi_1^*(s;x,y;z)&=&\frac{\tau(s+1;x+[z^{-1}],y)}{\tau(s+1;x,y)}e^{-\xi(x,z)}z^{-s},\nonumber\\
\Psi_2^*(s;x,y;z)&=&\frac{\tau(s;x,y+[z])}{\tau(s+1;x,y)}e^{-\xi(y,z^{-1})}z^{-s},\label{tauexpression}
\end{eqnarray}
where $[z]=(z,\frac{1}{2}z^2,\frac{1}{3}z^3,\cdots)$.

Another important object is the vertex operators \cite{ASM95} for
the Toda lattice hierarchy, which are defined in the following way
\begin{equation}\label{vertexoperators}
\mathbb{X}(x,\lambda,\mu):=\Big(\Big(\frac{\mu}{\lambda}\Big)^nX(x,\lambda,\mu)\Big)_{n\in\mathbb{Z}},
\quad\tilde{\mathbb{X}}(y,\lambda,\mu):=\Big(\Big(\frac{\lambda}{\mu}\Big)^nX(y,\lambda,\mu)\Big)_{n\in\mathbb{Z}},
\end{equation}
where
\begin{eqnarray}
    X(x,\lambda,\mu)&\equiv&\frac{1}{\lambda}:e^{\theta(\lambda)}::e^{\theta(\mu)}:=\frac{1}{\lambda}e^{\xi(x+[\lambda^{-1}],\mu)-\xi(x,\lambda)}
    e^{\sum_1^\infty\frac{1}{l}(\lambda^{-l}-\mu^{-l})\frac{\partial}{\partial
    x_l}}\nonumber\\
    &=&-\frac{1}{\mu}e^{\xi(x,\mu)-\xi(x-[\mu^{-1}],\lambda)}
    e^{\sum_1^\infty\frac{1}{l}(\lambda^{-l}-\mu^{-l})\frac{\partial}{\partial
    x_l}}+\delta(\lambda,\mu),
\end{eqnarray}
and
\begin{eqnarray}
\theta(\lambda)&\equiv&-\sum_{l=1}^\infty\lambda^lx_l+\sum_{l=1}^\infty\frac{1}{l}\lambda^{-l}\frac{\partial}{\partial
    x_l},\label{theta}
\end{eqnarray}
the columns $:...:$ indicate Wick normal ordering with respect to
the creation/annihilation``modes" $x_l$ and
$\frac{\partial}{\partial x_l}$, respectively. And the
delta-function is defined as
\begin{eqnarray}
\delta(\lambda,\mu)&\equiv&\sum_{n=-\infty}^\infty\frac{\mu^{n}}{\lambda^{n+1}}
=\frac{1}{\lambda}\frac{1}{1-\frac{\mu}{\lambda}}+\frac{1}{\mu}\frac{1}{1-\frac{\lambda}{\mu}}.\label{delta}
\end{eqnarray}
According to (\ref{tauexpression}), one can find
\begin{eqnarray}
\frac{\mathbb{X}(x,\lambda,\mu)\tau}{\tau}&=&\Big(\frac{1}{\lambda^2}\Psi_1(n;x+[\lambda^{-1}],y;\mu)\Psi_1^*(n-1;x,y;\lambda)\Big)_{n\in\mathbb{Z}}\nonumber\\
&=&\Big(-\frac{1}{\lambda\mu}\Psi_1(n;x,y;\mu)\Psi_1^*(n-1;x-[\mu^{-1}],y;\lambda)+\delta(\lambda,\mu)\Big)_{n\in\mathbb{Z}},\label{vertexwave1}\\
\frac{\tilde{\mathbb{X}}(y,\lambda,\mu)\tau}{\tau}&=&\Big(\frac{1}{\mu}\Psi_2(n-1;x,y+[\lambda^{-1}];\mu^{-1})\Psi_2^*(n-1;x,y;\lambda^{-1})\Big)_{n\in\mathbb{Z}}\nonumber\\
&=&\Big(-\frac{1}{\mu}\Psi_2(n;x,y;\mu^{-1})\Psi_2^*(n;x,y-[\mu^{-1}];\lambda^{-1})+\delta(\lambda,\mu)\Big)_{n\in\mathbb{Z}}.\label{vertexwave2}
\end{eqnarray}

The additional symmetry for the Toda lattice hierarchy
\cite{ASM95,takasaki1996,cheng2011} can be expressed in terms of the
Orlov-Shulman operators \cite{OS86}, which is defined as
\begin{equation}\label{osoperator}
    M\equiv(M_1,M_2)=W(\varepsilon,\varepsilon^*)W^{-1},
\end{equation}
where
$$\varepsilon=\rm{diag}[s]\Lambda^{-1},\quad \varepsilon^*=-\varepsilon^T
+\Lambda,$$
 satisfying
\begin{eqnarray}
M\Psi=(\pa_z,\pa_{z^{-1}})\Psi,&&[L,M]=(1,1), \nonumber\\
\pa_{x_n}M=[(L_1^n,0)_+,M],&&\pa_{y_n}M=[(0,L_2^n)_+,M].
\label{osproperty}
\end{eqnarray}
By introducing additional independent variables $x_{m,l}^*$ and
$y_{m,l}^*$, the actions of the additional symmetry on the wave
matrices are showed as
\begin{equation}\label{additionalsymactonwavematrices}
    \pa_{x_{m,l}^*}W=-(M_1^mL_1^l,0)_-W,\quad
    \pa_{y_{m,l}^*}W=-(0,M_2^mL_2^l)_-W.
\end{equation}
One can further find
\begin{eqnarray}
\pa_{x_{m,l}^*}\Psi=-(M_1^mL_1^l,0)_-\Psi,&& \pa_{y_{m,l}^*}\Psi=-(0,M_2^mL_2^l)_-\Psi,\nonumber\\
\pa_{x_{m,l}^*}L=[-(M_1^mL_1^l,0)_-,L],&& \pa_{y_{m,l}^*}L=[-(0,M_2^mL_2^l)_-,L],\nonumber\\
\pa_{x_{m,l}^*}M=[-(M_1^mL_1^l,0)_-,M],&&
\pa_{y_{m,l}^*}M=[-(0,M_2^mL_2^l)_-,M],\label{addsymactonothers}
\end{eqnarray}
and by acting on the space of the wave matrices, $\pa_{x_{m,l}^*}$
and $\pa_{y_{m,l}^*}$ forms into Lie algebra $w_\infty\times
w_\infty$.
\section{The squared eigenfunction symmetry for the Toda lattice hierarchy}

Given a couple of vector (adjoint) eigenfunctions $q$ and $r$,
\textbf{the squared eigenfunction flow} of the Toda lattice
hierarchy can be defined by its actions on the wave operators,
\begin{equation}\label{gsymw}
\pa_\alpha W_1=(q\otimes r)_l W_1,\quad \pa_\alpha W_2=-(q\otimes
r)_u W_2,
\end{equation}
or
\begin{equation}\label{gsymwnew}
\pa_\alpha W=(q\otimes r,0)_-W=-(0,q\otimes r)_-W.
\end{equation}
Note that for a matrix $P$, $P_u$ and $P_l$ denote the upper
(including diagonal) and strictly lower triangular parts of $P$
respectively, and $(A\otimes B)_{ij}=A_iB_j$.

 According to
(\ref{laxwaveexpression}), one can further have the squared
eigenfunction flow on the Lax operator
\begin{equation}\label{gsymlax}
\pa_\alpha L_1=[(q\otimes r)_l,L_1],\quad \pa_\alpha L_2=-[(q\otimes
r)_u, L_2],
\end{equation}
or
\begin{equation}\label{gsymlaxnew}
\pa_\alpha L=[(q\otimes r,0)_-,L]=-[(0,q\otimes r)_-,L].
\end{equation}

The proposition below shows that this squared eigenfunction flow is
indeed a kind of symmetry, and thus is called \textbf{the squared
eigenfunction symmetry}
\begin{proposition}
\begin{eqnarray}
[\pa_\alpha,\pa_{x_n}]=[\pa_\alpha,\pa_{y_n}]=0.
\end{eqnarray}
\end{proposition}
\begin{proof}
In fact, according to (\ref{tlhierarchy}),
(\ref{wavematricesequation}), (\ref{gsymw}) and (\ref{gsymlax})
\begin{eqnarray*}
[\pa_\alpha,\pa_{x_n}]W_1
&=&\pa_\alpha\Big((L_1^n)_uW_1\Big)-\pa_{x_n}\Big((q\otimes r)_lW_1\Big)\\
&=&[(q\otimes r)_l,L_1^n]_uW_1+(L_1^n)_u(q\otimes
r)_lW_1\\
&&-((L_1^n)_uq\otimes r)_lW_1+(q\otimes (L_1^n)_u^Tr)_lW_1-(q\otimes
r)_l(L_1^n)_uW_1\\
&=&[(q\otimes r)_l,L_1^n]_uW_1+[(L_1^n)_u,(q\otimes
r)_l]W_1\\
&&-\Big((L_1^n)_u(q\otimes r)\Big)_lW_1+\Big((q\otimes r)(L_1^n)_u\Big)_lW_1\\
&=&[(q\otimes r)_l,L_1^n]_uW_1+[(L_1^n)_u,(q\otimes
r)_l]W_1-[(L_1^n)_u,q\otimes r]_lW_1\\
&=&[(q\otimes r)_l,(L_1^n)_u]_uW_1+[(L_1^n)_u,(q\otimes
r)_l]W_1-[(L_1^n)_u,(q\otimes r)_l]_lW_1\\
&=&[(q\otimes r)_l,(L_1^n)_u]W_1+[(L_1^n)_u,(q\otimes r)_l]W_1=0.
\end{eqnarray*}
Note that $((L_1^n)_u)q\otimes r=(L_1^n)_u(q\otimes r)$, and
$q\otimes (L_1^n)_u^Tr=(q\otimes r)(L_1^n)_u$ are used in the third
identity. While $[A_u,B_u]_l=[A_l,B_l]_u=0$ is used in the fifth
identity.

Similarly,
$[\pa_\alpha,\pa_{x_n}]W_2=[\pa_\alpha,\pa_{y_n}]W_1=[\pa_\alpha,\pa_{y_n}]W_2=0$
can be proved.
\end{proof}
Define the following double expansions
\begin{eqnarray*}
Y_1(\lambda,\mu)&=&\sum_{m=0}^\infty
\frac{(\mu-\lambda)^m}{m!}\sum_{l=-\infty}^\infty\lambda^{-l-m-1}M_1^mL_1^{m+l},\\
Y_2(\lambda,\mu)&=&\sum_{m=0}^\infty
\frac{(\mu-\lambda)^m}{m!}\sum_{l=-\infty}^\infty\lambda^{-l-m-1}M_2^mL_2^{m+l},
\end{eqnarray*}
which can be viewed as the generator of the additional symmetries
for the Toda lattice hierarchy. This double expansions can be
related with the (adjoint) eigenfunctions in the following way,
\begin{proposition}\label{propaddisymmtoda2}

\begin{eqnarray}
Y_1(\lambda,\mu)&=&\lambda^{-1}\Psi_1(x,y;\mu)\otimes\Psi_1^*(x,y;\lambda),\label{generatortoda1}\\
Y_2(\lambda,\mu)&=&\lambda^{-1}\Psi_2(x,y;\mu^{-1})\otimes\Psi_2^*(x,y;\lambda^{-1}).\label{generatortoda2}
\end{eqnarray}
\end{proposition}
\begin{proof}
Firstly, according to Lemma \ref{lemma}, one has
\begin{eqnarray*}
M_1^mL_1^{m+l}&=& M_1^mW_1\Lambda^{m+l}W_1^{-1} \\
&=&{\rm res}_{z}z^{-1}\big(M_1^mW_1\Lambda^{m+l}\chi(z)\big)\otimes\big((W_1^{-1})^T\chi^*(z)\big)\\
&=&{\rm
res}_{z}\big(z^{-1+m+l}\pa_z^m\Psi_1(x,y;z)\big)\otimes\Psi_1^*(x,y;z),
\end{eqnarray*}
then,
\begin{eqnarray*}
Y_1(\lambda,\mu)&=&{\rm res}_{z}\sum_{m=0}^\infty \sum_{l=-\infty}^\infty\frac{z^{m+l}}{\lambda^{l+m+1}}\frac{(\mu-\lambda)^m}{m!}\big(z^{-1}\pa_z^m\Psi_1(x,y;z)\big)\otimes\Psi_1^*(x,y;z)\\
&=&{\rm res}_{z}\delta(\lambda,z)z^{-1}e^{(\mu-\lambda)\pa_z}\Psi_1(x,y;z)\otimes\Psi_1^*(x,y;z)\\
&=&\lambda^{-1}e^{(\mu-\lambda)\pa_\lambda}\Psi_1(x,y;\lambda)\otimes\Psi_1^*(x,y;\lambda)\\
&=&\lambda^{-1}\Psi_1(x,y;\mu)\otimes\Psi_1^*(x,y;\lambda),
\end{eqnarray*}
where ${\rm res}_z(\delta(\lambda,z)f(z))=f(\lambda)$ is used.

Thus (\ref{generatortoda1}) are proved.

Similarly for (\ref{generatortoda2}),
\begin{eqnarray*}
M_2^mL_2^{m+l}&=& M_2^mW_2\Lambda^{-m-l}W_2^{-1} \\
&=&{\rm res}_{z}z^{-1}\big(M_2^mW_2\Lambda^{-m-l}\chi(z^{-1})\big)\otimes\big((W_2^{-1})^T\chi^*(z^{-1})\big)\\
&=&{\rm
res}_{z}\big(z^{-1+m+l}\pa_z^m\Psi_2(x,y;z^{-1})\big)\otimes\Psi_2^*(x,y;z^{-1}),
\end{eqnarray*}
then
\begin{eqnarray*}
Y_2(\lambda,\mu)&=&res_{z}\sum_{m=0}^\infty \sum_{l=-\infty}^\infty\frac{z^{m+l}}{\lambda^{l+m+1}}\frac{(\mu-\lambda)^m}{m!}\big(z^{-1}\pa_z^m\Psi_2(x,y;z^{-1})\big)\otimes\Psi_2^*(x,y;z^{-1})\\
&=&res_{z}\delta(\lambda,z)z^{-1}e^{(\mu-\lambda)\pa_z}\Psi_2(x,y;z^{-1})\otimes\Psi_2^*(x,y;z^{-1})\\
&=&\lambda^{-1}e^{(\mu-\lambda)\pa_\lambda}\Psi_2(x,y;\lambda^{-1})\otimes\Psi_2^*(x,y;\lambda^{-1})\\
&=&\lambda^{-1}\Psi_2(x,y;\mu^{-1})\otimes\Psi_2^*(x,y;\lambda^{-1}).
\end{eqnarray*}
\end{proof}
If define $\pa_{\alpha_1}$ and $\pa_{\alpha_2}$ flows,
\begin{eqnarray}
\pa_{\alpha_1}
W_1&=&(\lambda^{-1}\Psi_1(x,y;\mu)\otimes\Psi_1^*(x,y;\lambda))_l
W_1,\\
\pa_{\alpha_1}
W_2&=&-(\lambda^{-1}\Psi_1(x,y;\mu)\otimes\Psi_1^*(x,y;\lambda))_u
W_2,
\end{eqnarray}
and
\begin{eqnarray}
\pa_{\alpha_2}
W_1&=&-(\lambda^{-1}\Psi_2(x,y;\mu^{-1})\otimes\Psi_2^*(x,y;\lambda^{-1}))_l
W_1,\\
\pa_{\alpha_2}
W_2&=&(\lambda^{-1}\Psi_2(x,y;\mu^{-1})\otimes\Psi_2^*(x,y;\lambda^{-1}))_u
W_2,
\end{eqnarray}
one can find that $\pa_{\alpha_1}$ and $\pa_{\alpha_2}$ are the
squared eigenfunction symmetries generated by the pair of
$\Psi_1(x,y;\mu)$ \& $\lambda^{-1}\Psi_1^*(x,y;\lambda)$ and the
pair of $-\Psi_2(x,y;\mu^{-1})$ \&
$\lambda^{-1}\Psi_2^*(x,y;\lambda^{-1})$ respectively.

Further from (\ref{generatortoda1}) and (\ref{generatortoda2}), it
is can be known that the squared eigenfunction symmetries
$\pa_{\alpha_1}$ and $\pa_{\alpha_2}$ are the generators of the
additional symmetries for the Toda lattice hierarchy, that is,
\begin{eqnarray*}
\pa_{\alpha_1}&=&\sum_{m=0}^\infty
\frac{(\mu-\lambda)^m}{m!}\sum_{k=-\infty}^\infty\lambda^{-k-m-1}\pa_{x_{m,m+k}^*},\\
\pa_{\alpha_2}&=&\sum_{m=0}^\infty
\frac{(\mu-\lambda)^m}{m!}\sum_{k=-\infty}^\infty\lambda^{-k-m-1}\pa_{y_{m,m+k}^*}.
\end{eqnarray*}

In the rest of this paper, we will mainly study this two particular
squared eigenfunction symmetries $\pa_{\alpha_1}$ and
$\pa_{\alpha_2}$, and show their actions on the tau function of the
Toda lattice hierarchy.

\section{Fay-like Identities and some Important Relations about the Wave Functions}
In this section, the Fay-like identities for the Toda lattice
hierarchy are reviewed and some important relations about the wave
functions are derived, which is helpful for the research on the
squared eigenfunction symmetries.

The general Fay-like identities are investigated in \cite{AvM99}. In
this paper, only some particular cases are needed. The starting
point of Fay-like identities is the bilinear identity of the Toda
lattice hierarchy in terms of the tau function, which can be derived
by substituting (\ref{tauexpression}) into the bilinear identity
(\ref{wavebilinearcomponent})
\begin{eqnarray}
 &&{\rm res}_{z}\Big(\tau(s;x-[z^{-1}],y)\tau(s'+1;x'+[z^{-1}],y')z^{s-s'-1}e^{\xi(x,z)-\xi(x',z)}\Big)\nonumber\\
 &=&{\rm res}_{z}\Big(\tau(s+1;x,y-[z^{-1}])\tau(s';x',y'+[z^{-1}])z^{s'-s-1}e^{\xi(y,z)-\xi(y',z)}\Big).\label{taubilinear}
\end{eqnarray}
Starting from (\ref{taubilinear}), the first group of particular
Fay-like identities, showed in the following lemma, which also
appears in \cite{teo2006}, can be easily obtained with the help of
the following identities,
\begin{eqnarray}
&&\frac{1}{(1-zs_1)(1-zs_2)}=\Big(\frac{1}{(1-zs_2)}-\frac{1}{(1-zs_1)}\Big)\frac{1}{z(s_2-s_1)},\label{usefulformula1}\\
&&{\rm res}_{z}\Big(\Big(\sum_{n=-\infty}^\infty
a_n(\zeta)z^{-n}\Big)\frac{1}{1-z/\zeta}\Big)
=\zeta\Big(\sum_{n=1}^\infty
a_n(\zeta)z^{-n}\Big)\Big|_{z=\zeta}.\label{usefulformula2}
\end{eqnarray}
\begin{lemma}\label{lemmafay2}
\begin{eqnarray}
(1)&&s_1\tau(n;x-[s_1],y)\tau(n+1;x-[s_2],y)-
s_2\tau(n;x-[s_2],y)\tau(n+1;x-[s_1],y)\nonumber\\
&=&(s_1-s_2)\tau(n+1;x,y)\tau(n;x-[s_1]-[s_2],y),\label{fayxx2}\\
(2)&&s_1\tau(n+1;x,y-[s_1])\tau(n;x,y-[s_2])-s_2\tau(n+1;x,y-[s_2])\tau(n;x,y-[s_1])\nonumber\\
&=&(s_1-s_2)\tau(n;x,y)\tau(n+1;x,y-[s_1]-[s_2]),\label{fayyy2}\\
(3)&&\tau(n;x-[s_1],y)\tau(n;x,y-[s_2])-\tau(n;x,y)\tau(n;x-[s_1],y-[s_2])\nonumber\\
&=&s_1s_2\tau(n-1;x-[s_1],y)\tau(n+1;x,y-[s_2]).\label{fayxy2}
\end{eqnarray}
\end{lemma}

Another group of particular Fay-like identities can also be obtained
starting from (\ref{taubilinear}). In fact, by considering the
following four cases of (\ref{taubilinear}):

\textbf{Case I}: $s'=s=n-1$, $x'=x-[s_1]-[s_2]-[s_3]$, $y'=y$,

\textbf{Case II}: $s'=s=n$, $x'=x-[s_1]-[s_2]$, $y'=y-[s_3]$,

\textbf{Case III}: $s'=s+1=n+1$, $x'=x-[s_3]$, $y'=y-[s_1]-[s_2]$,

\textbf{Case IV}: $s'=s+2=n+1$, $x'=x$, $y'=y-[s_1]-[s_2]-[s_3]$,\\
and using (\ref{usefulformula1}) and (\ref{usefulformula2}), one can
obtain the following lemma.
\begin{lemma}\label{lemmafay3}
\begin{eqnarray}
(1)&&s_1(s_2-s_3)\tau(n;x-[s_1],y)\tau(n;x-[s_2]-[s_3],y)\nonumber\\
&+&s_2(s_3-s_1)\tau(n;x-[s_2],y)\tau(n;x-[s_3]-[s_1],y)\nonumber\\
&+&s_3(s_1-s_2)\tau(n;x-[s_3],y)\tau(n;x-[s_1]-[s_2],y)=0,\label{fayxx3}\\
(2)&&s_1\tau(n;x-[s_1],y)\tau(n+1;x-[s_2],y-[s_3])\nonumber\\
&-&s_2\tau(n;x-[s_2],y)\tau(n+1;x-[s_1],y-[s_3])\nonumber\\
&=&(s_1-s_2)\tau(n;x-[s_1]-[s_2],y)\tau(n+1;x,y-[s_3]),\label{fayxy3}\\
(3)&&(s_1-s_2)s_3\tau(n+1;x,y-[s_1]-[s_2])\tau(n-1;x-[s_3],y)\nonumber\\
&=&\tau(n;x,y-[s_1])\tau(n;x-[s_3],y-[s_2])-\tau(n;x,y-[s_2])\tau(n;x-[s_3],y-[s_1]),\label{fayyx3}\\
(4)&&(s_2-s_3)\tau(n;x,y-[s_1])\tau(n+1;x,y-[s_2]-[s_3])\nonumber\\
&+&(s_3-s_1)\tau(n;x,y-[s_2])\tau(n+1;x,y-[s_3]-[s_1])\nonumber\\
&+&(s_1-s_2)\tau(n;x,y-[s_3])\tau(n+1;x,y-[s_1]-[s_2])=0.\label{fayyy3}
\end{eqnarray}
\end{lemma}
Next with the help of the two groups of particular Fay-like
identities, one can get the following relations about the wave
functions for the Toda lattice hierarchy.
\begin{lemma}\label{lmm1}
\begin{eqnarray}
\Psi_1^*(n;x,y;\lambda)\Psi_1(n;x,y;z)&=&z^{-1}(\Psi_1(n+1;x,y;z)\Psi_1^*(n;x-[z^{-1}],y;\lambda)\nonumber\\
&&-\Psi_1(n;x,y;z)\Psi_1^*(n-1;x-[z^{-1}],y;\lambda)),\label{wave11}\\
\Psi_1^*(n;x,y;\lambda)\Psi_2(n;x,y;z)&=&\Psi_2(n;x,y;z)\Psi_1^*(n;x,y-[z];\lambda)\nonumber\\
&&-\Psi_2(n+1;x,y;z)\Psi_1^*(n+1;x,y-[z];\lambda),\label{wave12}\\
\Psi_2^*(n;x,y;\lambda^{-1})\Psi_1(n;x,y;z)&=&z^{-1}(\Psi_1(n+1;x,y;z)\Psi_2^*(n;x-[z^{-1}],y;\lambda^{-1})\nonumber\\
&&-\Psi_1(n;x,y;z)\Psi_2^*(n-1;x-[z^{-1}],y;\lambda^{-1})),\label{wave21}\\
\Psi_2^*(n;x,y;\lambda^{-1})\Psi_2(n;x,y;z)&=&\Psi_2(n;x,y;z)\Psi_2^*(n;x,y-[z];\lambda^{-1})\nonumber\\
&&-\Psi_2(n+1;x,y;z)\Psi_2^*(n+1;x,y-[z];\lambda^{-1}).\label{wave22}
\end{eqnarray}

\end{lemma}
\begin{proof}Firstly, according to (\ref{tauexpression}) and (\ref{fayxx2}), one can get
\begin{eqnarray*}
&&\Psi_1^*(n;x,y;\lambda)\Psi_1(n;x,y;z)\\
&=&e^{\xi(x,z)-\xi(x,\lambda)}\Big(\frac{z}{\lambda}\Big)^n
\frac{\tau(n+1;x+[\lambda^{-1}],y)\tau(n;x-[z^{-1}],y)}{\tau(n+1;x,y)\tau(n;x,y)}\\
&=&e^{\xi(x,z)-\xi(x,\lambda)}\Big(\frac{z}{\lambda}\Big)^n\Big(\frac{z^{-1}}{z^{-1}-\lambda^{-1}}\frac{\tau(n;x+[\lambda^{-1}]-[z^{-1}],y)}{\tau(n;x,y)}\\
&&-\frac{\lambda^{-1}}{z^{-1}-\lambda^{-1}}\frac{\tau(n+1;x+[\lambda^{-1}]-[z^{-1}],y)}{\tau(n+1;x,y)}\Big)\\
&=&z^{-1}(\Psi_1(n+1;x,y;z)\Psi_1^*(n;x-[z^{-1}],y;\lambda)-\Psi_1(n;x,y;z)\Psi_1^*(n-1;x-[z^{-1}],y;\lambda)),
\end{eqnarray*}
and thus (\ref{wave11}) is obtained.

Then, similarly (\ref{tauexpression}) and (\ref{fayxy2}) leads to
(\ref{wave12}) and (\ref{wave21}), that is,
\begin{eqnarray*}
&&\Psi_1^*(n;x,y;\lambda)\Psi_2(n;x,y;z)\\
&=&\Big(\frac{z}{\lambda}\Big)^ne^{-\xi(x,\lambda)+\xi(y,z^{-1})}\frac{\tau(n+1;x+[\lambda^{-1}],y)\tau(n+1;x,y-[z])}{\tau(n+1;x,y)\tau(n;x,y)}\\
&=&\Big(\frac{z}{\lambda}\Big)^ne^{-\xi(x,\lambda)+\xi(y,z^{-1})}\Big(\frac{\tau(n+1;x+[\lambda^{-1}],y-[z])}{\tau(n;x,y)}-\frac{z}{\lambda}\cdot\frac{\tau(n+2;x+[\lambda^{-1}],y-[z])}{\tau(n+1;x,y)}\Big)\\
&=&\Psi_2(n;x,y;z)\Psi_1^*(n;x,y-[z];\lambda)-\Psi_2(n+1;x,y;z)\Psi_1^*(n+1;x,y-[z];\lambda),
\end{eqnarray*}
and
\begin{eqnarray*}
&&\Psi_2^*(n;x,y;\lambda^{-1})\Psi_1(n;x,y;z)\\
&=&e^{\xi(x,z)-\xi(y,\lambda)}(\lambda z)^{n}\frac{\tau(n;x,y+[\lambda^{-1}])\tau(n;x-[z^{-1}],y)}{\tau(n+1;x,y)\tau(n;x,y)}\\
&=&e^{\xi(x,z)-\xi(y,\lambda)}(\lambda z)^{n}\Big(\frac{\tau(n;x-[z^{-1}],y+[\lambda^{-1}])}{\tau(n+1;x,y)}-(\lambda z)^{-1}\frac{\tau(n-1;x-[z^{-1}],y+[\lambda^{-1}])}{\tau(n;x,y)}\Big)\\
&=&z^{-1}(\Psi_1(n+1;x,y;z)\Psi_2^*(n;x-[z^{-1}],y;\lambda)-\Psi_1(n;x,y;z)\Psi_2^*(n-1;x-[z^{-1}],y;\lambda)).
\end{eqnarray*}
At last, by using (\ref{tauexpression}) and (\ref{fayyy2}), one can
obtain (\ref{wave22}).
\begin{eqnarray*}
&&\Psi_2^*(n;x,y;\lambda^{-1})\Psi_2(n;x,y;z)\\
&=&e^{\xi(y,z^{-1})-\xi(y,\lambda)}(\lambda z)^n\frac{\tau(n;x,y+[\lambda^{-1}])\tau(n+1;x,y-[z])}{\tau(n+1;x,y)\tau(n;x,y)}\\
&=&\lambda^{-1}(\lambda^{-1}-z)^{-1}e^{\xi(y,z^{-1})-\xi(y,\lambda)}(\lambda z)^n\Big(\frac{\tau(n;x,y+[\lambda^{-1}]-[z])}{\tau(n;x,y)}-(\lambda z)\cdot\frac{\tau(n+1;x,y+[\lambda^{-1}]-[z])}{\tau(n+1;x,y)}\Big)\\
&=&\Psi_2(n;x,y;z)\Psi_2^*(n;x,y-[z];\lambda^{-1})-\Psi_2(n+1;x,y;z)\Psi_2^*(n+1;x,y-[z];\lambda^{-1}).
\end{eqnarray*}
\end{proof}
From (\ref{wave11})-(\ref{wave22}), one can further get
\begin{eqnarray}
&&\sum_{k<0}\Psi_1^*(n+k;x,y;\lambda)\Psi_1(n+k;x,y;z)
=z^{-1}\Psi_1^*(n-1;x-[z^{-1}],y;\lambda)\Psi_1(n;x,y;z),\label{sumwave11}\\
&&\sum_{k\geq0}\Psi_1^*(n+k;x,y;\lambda)\Psi_2(n+k;x,y;z)
=\Psi_2(n;x,y;z)\Psi_1^*(n;x,y-[z];\lambda)\label{sumwave12},\\
&&\sum_{k<0}\Psi_2^*(n+k;x,y;\lambda^{-1})\Psi_1(n+k;x,y;z)
=z^{-1}\Psi_1(n;x,y;z)\Psi_2^*(n-1;x-[z^{-1}],y;\lambda^{-1}),\label{sumwave21}\\
&&\sum_{k\geq0}\Psi_2^*(n+k;x,y;\lambda^{-1})\Psi_2(n+k;x,y;z)
=\Psi_2(n;x,y;z)\Psi_2^*(n;x,y-[z];\lambda^{-1}).\label{sumwave22}
\end{eqnarray}

If introduce the following two notations,
\begin{equation}\label{shiftoperator}
    G_1(\xi)f(x,y;z)=f(x-[\xi^{-1}],y;z),\quad
    G_2(\xi)f(x,y;z)=f(x,y-[\xi];z),
\end{equation}
one can further have another group of relations about the wave
functions, which is the lemma below.
\begin{lemma}\label{lmm2}

\begin{eqnarray}
&&(G_1(z)-1)\Psi_1(n;x,y;\mu)\Psi_1^*(n-1;x-[\mu^{-1}],y;\lambda)\nonumber\\
&=&-\frac{\mu}{z}\Psi_1(n;x,y;\mu)\Psi_1^*(n-1;x-[z^{-1}],y;\lambda),\label{shiftwave11}\\
&&(\Lambda G_2(z)-1)\Psi_1(n;x,y;\mu)\Psi_1^*(n-1;x-[\mu^{-1}],y;\lambda)\nonumber\\
&=&\mu\Psi_1(n;x,y;\mu)\Psi_1^*(n;x,y-[z];\lambda),\label{shiftwave12}\\
&&(G_1(z)-1)\Psi_2(n;x,y;\mu^{-1})\Psi_2^*(n;x,y-[\mu^{-1}];\lambda^{-1})\nonumber\\
&=&z^{-1}\Psi_2(n;x,y;\mu^{-1})\Psi_2^*(n-1;x-[z^{-1}],y;\lambda^{-1}),\label{shiftwave21}\\
&&(\Lambda
G_2(z)-1)\Psi_2(n;x,y;\mu^{-1})\Psi_2^*(n;x,y-[\mu^{-1}];\lambda^{-1})\nonumber\\
&=&-\Psi_2(n;x,y;\mu^{-1})\Psi_2^*(n;x,y-[z];\lambda^{-1}).\label{shiftwave22}
\end{eqnarray}

\end{lemma}
\begin{proof}
Firstly for (\ref{shiftwave11}), with the help of
(\ref{tauexpression}) and (\ref{shiftoperator}), one has
\begin{eqnarray*}
&&(G_1(z)-1)\Psi_1(n;x,y;\mu)\Psi_1^*(n-1;x-[\mu^{-1}],y;\lambda)\\
&=&\lambda\Big(1-\frac{\lambda}{\mu}\Big)^{-1}\Big(\frac{\mu}{\lambda}\Big)^ne^{\xi(x,\mu)-\xi(x,\lambda)}z^{-1}\Big(1-\frac{\lambda}{z}\Big)^{-1}\frac{1}{\tau(n;x,y)\tau(n;x-[z^{-1}],y)}\\
&&\times\big((z-\mu)\tau(n;x,y)\tau(n;x+[\lambda^{-1}]-[\mu^{-1}]-[z^{-1}],y)\\
&&-(z-\lambda)\tau(n;x-[z^{-1}],y)\tau(n;x+[\lambda^{-1}]-[\mu^{-1}],y)\big),
\end{eqnarray*}
and further according to (\ref{fayxx3}) with $s_1=\lambda^{-1}$,
$s_2=\mu^{-1}$, and $s_3=z^{-1}$,
\begin{eqnarray*}
&&(G_1(z)-1)\Psi_1(n;x,y;\mu)\Psi_1^*(n-1;x-[\mu^{-1}],y;\lambda)\nonumber\\
&=&-\frac{\lambda\mu}{z}\Big(\frac{\mu}{\lambda}\Big)^ne^{\xi(x,\mu)-\xi(x,\lambda)}\Big(1-\frac{\lambda}{z}\Big)^{-1}\frac{\tau(n;x-[\mu^{-1}],y)\tau(n;x+[\lambda^{-1}]-[z^{-1}],y)}{\tau(n;x,y)\tau(n;x-[z^{-1}],y)}\nonumber\\
&=&-\frac{\mu}{z}\Psi_1(n;x,y;\mu)\Psi_1^*(n-1;x-[z^{-1}],y;\lambda).
\end{eqnarray*}

Then by (\ref{tauexpression}), (\ref{shiftoperator}) and
(\ref{fayxy3}) with $s_1=\lambda^{-1}$, $s_2=\mu^{-1}$, and $s_3=z$,
(\ref{shiftwave12}) is obtained.
\begin{eqnarray*}
&&(\Lambda G_2(z)-1)\Psi_1(n;x,y;\mu)\Psi_1^*(n-1;x-[\mu^{-1}],y;\lambda)\\
&=&\Big(1-\frac{\lambda}{\mu}\Big)^{-1}\Big(\frac{\mu}{\lambda}\Big)^ne^{\xi(x,\mu)-\xi(x,\lambda)}\frac{1}{\tau(n;x,y)\tau(n+1;x,y-[z])}\\
&&\times\big(\mu\tau(n;x,y)\tau(n+1;x+[\lambda^{-1}]-[\mu^{-1}],y-[z])\\
&&-\lambda\tau(n+1;x,y-[z])\tau(n;x+[\lambda^{-1}]-[\mu^{-1}],y)\big)\\
&=&\mu\Big(\frac{\mu}{\lambda}\Big)^ne^{\xi(x,\mu)-\xi(x,\lambda)}\frac{\tau(n;x-[\mu^{-1}],y)\tau(n+1;x+[\lambda^{-1}],y-[z])}{\tau(n;x,y)\tau(n+1;x,y-[z])}\\
&=&\mu\Psi_1(n;x,y;\mu)\Psi_1^*(n;x,y-[z];\lambda).
\end{eqnarray*}

Similarly, (\ref{tauexpression}), (\ref{shiftoperator}) and
(\ref{fayyx3}) for $s_1=\lambda^{-1},s_2=\mu^{-1},s_3=z^{-1}$ lead
to (\ref{shiftwave21})
\begin{eqnarray*}
&&(G_1(z)-1)\Psi_2(n;x,y;\mu^{-1})\Psi_2^*(n;x,y-[\mu^{-1}];\lambda^{-1})\nonumber\\
&=&e^{\xi(y,\mu)-\xi(y,\lambda)}\Big(1-\frac{\lambda}{\mu}\Big)^{-1}\Big(\frac{\lambda
}{\mu
}\Big)^n\frac{1}{\tau(n;x,y)\tau(n;x-[z^{-1}],y)}\times\nonumber\\
&&\Big(\tau(n;x-[z^{-1}],y+[\lambda^{-1}]-[\mu^{-1}])\tau(n;x,y)
-\tau(n;x,y+[\lambda^{-1}]-[\mu^{-1}])\tau(n;x-[z^{-1}],y)\Big)\\
&=&\lambda^{-1}z^{-1}e^{\xi(y,\mu)-\xi(y,\lambda)}\Big(\frac{\lambda
}{\mu
}\Big)^n\frac{\tau(n+1;x,y-[\mu^{-1}])\tau(n-1;x-[z^{-1}],y+[\lambda^{-1}])}{\tau(n;x,y)\tau(n;x-[z^{-1}],y)}\nonumber\\
&=&z^{-1}\Psi_2(n;x,y;\mu^{-1})\Psi_2^*(n-1;x-[z^{-1}],y;\lambda^{-1}).
\end{eqnarray*}

At last,  (\ref{shiftwave22}) can be got by according to
(\ref{fayyy3}) with
$s_1=\lambda^{-1},s_2=\mu^{-1},s_3=z$,
\begin{eqnarray*}
&&(\Lambda G_2(z)-1)\Psi_2(n;x,y;\mu^{-1})\Psi_2^*(n;x,y-[\mu^{-1}];\lambda^{-1})\\
&=&\Big(1-\frac{\lambda}{\mu}\Big)^{-1}\Big(\frac{\lambda}{\mu
}\Big)^ne^{\xi(y,\mu)-\xi(y,\lambda)}\frac{(\lambda^{-1}-z)^{-1}}{\tau(n;x,y)\tau(n+1;x,y-[z])}\\
&&\times\big((\mu^{-1}-z)\tau(n+1;x,y+[\lambda^{-1}]-[\mu^{-1}]-[z])\tau(n;x,y)\nonumber\\
&&-(\lambda^{-1}-z)\tau(n;x,y+[\lambda^{-1}]-[\mu^{-1}])\tau(n+1;x,y-[z])\big)\\
&=&-\lambda^{-1}(\lambda^{-1}-z)^{-1}\Big(\frac{\lambda}{\mu
}\Big)^ne^{\xi(y,\mu)-\xi(y,\lambda)}\frac{\tau(n+1;x,y-[\mu^{-1}])\tau(n;x,y+[\lambda^{-1}]-[z])}{\tau(n;x,y)\tau(n+1;x,y-[z])}\\
&=&-\Psi_2(n;x,y;\mu^{-1})\Psi_2^*(n;x,y-[z];\lambda^{-1}).
\end{eqnarray*}
\end{proof}

\section{the Actions of the Squared Eigenfunction Symmetries on the Tau Function}
Based upon the preparation above, we now give the actions of the
squared eigenfunction symmetry on the wave functions, and further on
the tau function.
\begin{proposition}
The actions of the squared eigenfunction symmetry on the wave
functions is given as follows
\begin{eqnarray}
\frac{\pa_{\alpha_1}\Psi}{\Psi}&=&\Big((G_1(z)-1)\frac{\mathbb{X}(x,\lambda,\mu)\tau}{\tau},(\Lambda G_2(z)-1)
\frac{\mathbb{X}(x,\lambda,\mu)\tau}{\tau}\Big),\label{ghostwave1}\\
\frac{\pa_{\alpha_2}\Psi}{\Psi}&=&\frac{\mu}{\lambda}\Big((G_1(z)-1)\frac{\tilde
{\mathbb{X}}(y,\lambda,\mu)\tau}{\tau},(\Lambda
G_2(z)-1)\frac{\tilde
{\mathbb{X}}(y,\lambda,\mu)\tau}{\tau}\Big).\label{ghostwave2}
\end{eqnarray}
\end{proposition}
\begin{proof}
Firstly by (\ref{sumwave11}), (\ref{shiftwave11}) and
(\ref{vertexwave1}),
\begin{eqnarray*}
\pa_{\alpha_1}\Psi_1&=&(\lambda^{-1}\Psi_1(x,y;\mu)\otimes\Psi_1^*(x,y;\lambda))_l\Psi_1(x,y;z)\\
&=&\left(\lambda^{-1}\Psi_1(n;x,y;\mu)\sum_{k<0}^\infty\Psi_1^*(n+k;x,y;\lambda)\Psi_1(n+k;x,y;z)\right)_{n\in\mathbb{Z}}\\
&=&\left(\lambda^{-1}z^{-1}\Psi_1(n;x,y;\mu)\Psi_1^*(n-1;x-[z^{-1}],y;\lambda)\Psi_1(n;x,y;z)\right)_{n\in\mathbb{Z}}\\
&=&\left(-\lambda^{-1}\mu^{-1}\Psi_1(n;x,y;z)(G_1(z)-1)\big(\Psi_1(n;x,y;\mu)\Psi_1^*(n-1;x-[\mu^{-1}],y;\lambda)\big)\right)_{n\in\mathbb{Z}}\\
&=&\Psi_1(G_1(z)-1)\frac{\mathbb{X}(x,\lambda,\mu)\tau}{\tau}.
\end{eqnarray*}

Then with the help of (\ref{sumwave12}), (\ref{shiftwave12}) and
(\ref{vertexwave1}),
\begin{eqnarray*}
\pa_{\alpha_1}\Psi_2&=&-(\lambda^{-1}\Psi_1(x,y;\mu)\otimes\Psi_1^*(x,y;\lambda))_u\Psi_2(x,y;z)\\
&=&-\left(\lambda^{-1}\Psi_1(n;x,y;\mu)\sum_{k\geq0}^\infty\Psi_1^*(n+k;x,y;\lambda)\Psi_2(n+k;x,y;z)\right)_{n\in\mathbb{Z}}\\
&=&-\left(\lambda^{-1}\Psi_1(n;x,y;\mu)\Psi_1^*(n;x,y-[z];\lambda)\Psi_2(n;x,y;z)\right)_{n\in\mathbb{Z}}\\
&=&\left(-\lambda^{-1}\mu^{-1}\Psi_2(n;x,y;z)(\Lambda G_2(z)-1)\big(\Psi_1(n;x,y;\mu)\Psi_1^*(n-1;x-[\mu^{-1}],y;\lambda)\big)\right)_{n\in\mathbb{Z}}\\
&=&\Psi_2(\Lambda
G_2(z)-1)\frac{\mathbb{X}(x,\lambda,\mu)\tau}{\tau},
\end{eqnarray*}
and (\ref{sumwave21}), (\ref{shiftwave21}) and (\ref{vertexwave2})
lead to
\begin{eqnarray*}
\pa_{\alpha_2}\Psi_1&=&-(\lambda^{-1}\Psi_2(x,y;\mu^{-1})\otimes\Psi_2^*(x,y;\lambda^{-1}))_l\Psi_1(x,y;z)\\
&=&-\left(\lambda^{-1}\Psi_2(n;x,y;\mu^{-1})\sum_{k<0}^\infty\Psi_2^*(n+k;x,y;\lambda^{-1})\Psi_1(n+k;x,y;z)\right)_{n\in\mathbb{Z}}\\
&=&-\left(\lambda^{-1}z^{-1}\Psi_2(n;x,y;\mu^{-1})\Psi_2^*(n-1;x-[z^{-1}],y;\lambda^{-1})\Psi_1(n;x,y;z)\right)_{n\in\mathbb{Z}}\\
&=&\left(-\lambda^{-1}\Psi_1(n;x,y;z)(G_1(z)-1)\big(\Psi_2(n;x,y;\mu^{-1})\Psi_2^*(n;x,y-[\mu^{-1}];\lambda^{-1})\big)\right)_{n\in\mathbb{Z}}\\
&=&\frac{\mu}{\lambda}\Psi_1(G_1(z)-1)\frac{\tilde
{\mathbb{X}}(y,\lambda,\mu)\tau}{\tau}.
\end{eqnarray*}

At last, according to (\ref{sumwave22}), (\ref{shiftwave22}) and
(\ref{vertexwave2})
\begin{eqnarray*}
\pa_{\alpha_2}\Psi_2&=&(\lambda^{-1}\Psi_2(x,y;\mu^{-1})\otimes\Psi_2^*(x,y;\lambda^{-1}))_u\Psi_2(x,y;z)\\
&=&\left(\lambda^{-1}\Psi_2(n;x,y;\mu^{-1})\sum_{k\geq0}^\infty\Psi_2^*(n+k;x,y;\lambda^{-1})\Psi_2(n+k;x,y;z)\right)_{n\in\mathbb{Z}}\\
&=&\left(\lambda^{-1}\Psi_2(n;x,y;\mu^{-1})\Psi_2^*(n-1;x,y-[z];\lambda^{-1})\Psi_2(n;x,y;z)\right)_{n\in\mathbb{Z}}\\
&=&\left(-\lambda^{-1}\Psi_2(n;x,y;z)(\Lambda G_2(z)-1)\big(\Psi_2(n;x,y;\mu^{-1})\Psi_2^*(n;x,y-[\mu^{-1}];\lambda^{-1})\big)\right)_{n\in\mathbb{Z}}\\
&=&\frac{\mu}{\lambda}\Psi_2(\Lambda G_2(z)-1)\frac{\tilde
{\mathbb{X}}(y,\lambda,\mu)\tau}{\tau}.
\end{eqnarray*}
\end{proof}
At last according to the relations (\ref{tauexpression}) between the
wave functions and the tau functions, we get the following
proposition.
\begin{proposition}
The actions of the squared eigenfunction symmetries on the tau
functions are given as follows
\begin{eqnarray}
\pa_{\alpha_1}\tau=\mathbb{X}(x,\lambda,\mu)\tau,\quad
\pa_{\alpha_2}\tau=\frac{\mu}{\lambda}\tilde
{\mathbb{X}}(y,\lambda,\mu)\tau. \label{ghosttau2}
\end{eqnarray}
\end{proposition}
\section{Conclusions and Discussions}
The squared eigenfunction symmetry of the Toda lattice hierarchy is
constructed explicitly in the form of the Kronecker product of the
vector eigenfunction and the vector adjoint eigenfunction. And the
relation with the additional symmetry is also investigated, that is,
the squared eigenfunction symmetry can be viewed as the generating
function of the additional symmetries when the eigenfunction and the
adjoint eigenfunction are the wave function and the adjoint wave
function respectively. The action of the particular squared
eigenfunction symmetry on the tau function of the Toda lattice
hierarchy is obtained with the help of the Fay-like identities and
the relations about the wave functions. And thus another proof of
the ASvM formulas for the Toda lattice hierarchy is given. The
squared eigenfunction symmetry here is expected to be applied in the
study of the symmetry constraints for the Toda lattice hierarchy.
And the spectral representation for the eigenfunction of the Toda
lattice hierarchy, similar to KP hierarchy case \cite{aratyn1998},
is also expected to be set up. We will focus on these questions
later in the future paper.

{\bf Acknowledgements}\\
This work is supported by ``the Fundamental Research Funds for the
Central Universities" No. 2012QNA45.

\end{document}